\documentclass[
final
]{dmtcs-episciences}

\newtheorem{theorem}{Theorem}
\newtheorem{property}{Property}
\newtheorem{definition}{Definition}
\newtheorem{lemma}{Lemma}

\usepackage[utf8]{inputenc}
\usepackage{subfigure}
\newcommand{\OPT}{\mathrm{OPT}}
\newcommand{\SOL}{\mathrm{SOL}}

\usepackage[round]{natbib}

\author{Eric Angel\affiliationmark{1}
  \and Evripidis Bampis\affiliationmark{2}
  \and Bruno Escoffier\affiliationmark{2}
  \and Michael Lampis\affiliationmark{3}}
\title[Parameterized Power Vertex Cover]{Parameterized Power Vertex Cover\thanks{This is an extended version of a work presented at the 42nd International Workshop on Graph-Theoretic Concepts in Computer Science (WG 2016) \cite{DBLP:conf/wg/AngelBEL16}.}}
\affiliation{
  IBISC, Universit\'{e} Evry Val d'Essone\\
  Sorbonne Universit\'{e}, CNRS, LIP6 UMR 7606\\
  CNRS UMR 7243 and Universit\'e Paris Dauphine}
\received{2018-2-1}

\revised{2018-5-29, 2018-9-21}

\accepted{2018-9-21}

\begin{document}
\publicationdetails{20}{2018}{2}{10}{4256}
\maketitle

\begin{abstract}

We study a recently introduced generalization  of the \textsc{Vertex Cover}
{\sc (VC)} problem, called \textsc{Power Vertex Cover} \textsc{(PVC)}. In this
problem, each edge of the input graph is supplied with a positive integer
\emph{demand}. A solution is an assignment of (power) values to the vertices,
so that for each edge one of its endpoints has value as high as the demand, and
the total sum of power values assigned is minimized.

We investigate how this generalization affects the parameterized complexity of \textsc{Vertex
Cover}.  On the
positive side, when parameterized by the value of the optimal $P$, we give an
$O^*(1.274^P)$-time branching algorithm ($O^*$ is used to hide factors polynomial in the input size), and also an $O^*(1.325^P)$-time algorithm for
the more general asymmetric case of the problem, where the demand of each edge
may differ for its two endpoints. When the parameter is the number of vertices
$k$ that receive positive value, we give $O^*(1.619^k)$ and $O^*(k^k)$-time
algorithms for the symmetric and asymmetric cases respectively, as well as a
simple quadratic kernel for the asymmetric case.

We also show that \textsc{PVC} becomes significantly harder than classical
\textsc{VC} when parameterized by the graph's treewidth $t$. More specifically,
we prove that unless the ETH is false, there is no $n^{o(t)}$-time algorithm for
\textsc{PVC}. We give a method to overcome this hardness by designing an FPT
\emph{approximation scheme} which gives a $(1+\epsilon)$-approximation to the
optimal solution in time FPT in parameters $t$ and $1/\epsilon$.

\end{abstract}

\section{Introduction}

In the classical {\sc Vertex Cover (VC)}  problem, we are given a graph
$G=(V,E)$ and we aim to find a minimum cardinality cover of the edges, i.e. a
subset of the vertices $C \subseteq V$ such that for every edge $e \in E$, at
least one of its endpoints belongs to $C$. {\sc Vertex Cover} is one of the
most extensively studied NP-hard problems in both approximation and
parameterized algorithms, see \cite{WS11,N06}.

In this paper, we study a natural generalization of the  {\sc VC} problem,
which we call \textsc{Power Vertex Cover (PVC)}. In this generalization,  we
are given an edge-weighted graph $G=(V,E)$ and we are asked to assign (power)
values to its vertices. We say that an edge $e$ is covered if at least one of
its endpoints is assigned a value greater than or equal to the weight of $e$.
The goal is to determine a valuation such that all edges are covered and the
sum of all values assigned is minimized.  Clearly, if all edge weights are
equal to 1, then this problem coincides with \textsc{VC}.

\textsc{Power Vertex Cover} was recently introduced in \cite{Angel2015},
motivated by practical applications in sensor networks (hence the term
``power''). The main question posed in \cite{Angel2015} was whether this more
general problem is harder to approximate than \textsc{VC}. It was then shown
that \textsc{PVC} retains enough of the desirable structure of \textsc{VC} to
admit a similar 2-approximation algorithm, even for the more general case where
the power needed to cover the edge $(u,v)$ is not the same for $u$ and $v$ (a
case referred to as \textsc{Directed Power Vertex Cover (DPVC)}).

The goal of this paper is to pose a similar question in the context of
parameterized complexity: is it possible to leverage known FPT results for
\textsc{VC} to obtain FPT algorithms for this more general version? We offer a
number of both positive and negative results. Specifically:

\begin{itemize}
\item When the parameter is the value of the optimal solution $P$ (and all weights
  are positive integers), we show an $O^*(1.274^P)$-time branching algorithm for
\textsc{PVC}, and an $O^*(1.325^P)$-time algorithm for \textsc{DPVC}. Thus, in this
case, the two problems behave similarly to classical \textsc{VC}.

\item When the parameter is the \emph{cardinality} $k$ of the optimal solution,
  that is, the number of vertices to be assigned non-zero values, we show
$O^*(1.619^k)$ and $O^*(k^k)$-time algorithms for \textsc{PVC} and \textsc{DPVC}
respectively, as well as a simple quadratic (vertex) kernel for \textsc{DPVC}, similar
to the classical Buss kernel for \textsc{VC}. This raises the question of
whether a kernel of order \emph{linear} in $k$ can be obtained. We give some
negative evidence in this direction, by showing that an LP-based approach is
very unlikely to succeed.  More strongly, we show that, given an optimal
\emph{fractional} solution to \textsc{PVC} which assigns value 0 to a vertex,
it is NP-hard to decide if an optimal solution exists that does the same.

\item When the parameter is the treewidth $t$ of the input graph, we show through
  an FPT reduction from \textsc{Clique} that there is no $n^{o(t)}$ algorithm
for \textsc{PVC} unless the ETH is false. This is essentially tight, since we
also supply an $O^*((\Delta+1)^t)$-time algorithm, where $\Delta$ is the maximum degree of the graph, and is in stark contrast to
\textsc{VC}, which admits an $O^*(2^t)$-time algorithm. We complement this hardness
result with an FPT approximation scheme, that is, an algorithm which, for any
$\epsilon>0$ returns a $(1+\epsilon)$-approximate solution while running in
time FPT in $t$ and $\frac{1}{\epsilon}$.  Specifically, our algorithm runs in
time $\left( O(\frac{\log n}{\epsilon})\right)^t n^{O(1)}$.  We also generalize the $O^*(2^t)$-time algorithm for \textsc{VC} to an $O^*((M+1)^t))$-time algorithm for \textsc{PVC} where $M$ is the maximum edge weight ($M=1$ for \textsc{VC}).  
\end{itemize}

Our results thus indicate that \textsc{PVC} occupies a very interesting spot in
terms of its parameterized complexity. On the one hand, \textsc{PVC} carries
over many of the desirable algorithmic properties of \textsc{VC}: branching
algorithms and simple kernelization algorithms can be directly applied. On the
other, this problem seems to be considerably harder in several (sometimes
surprising) respects. In particular, neither the standard treewidth-based DP
techniques, nor the Nemhauser-Trotter theorem can be applied to obtain results
comparable to those for \textsc{VC}. In fact, in the latter case, the existence
of edge weights turns a trivial problem (all vertices with fractional optimal
value 0 are placed in the independent set) to an NP-hard one. Yet, despite
its added hardness, \textsc{PVC} in fact admits an FPT approximation scheme, a
property that is at the moment known for only a handful of other W-hard
problems. Because of all these, we view the results of this paper as a first
step towards a deeper understanding of a natural generalization of \textsc{VC}
that merits further investigation.

\subsubsection*{Previous work}

As mentioned, \textsc{PVC} and \textsc{DPVC} were introduced in
\cite{Angel2015}, where 2-approximation algorithms were presented  for general
graphs and it was proved that, like \textsc{VC}, the problem can be solved in
polynomial time for bipartite graphs.

\textsc{Vertex Cover} is one of the most studied problems in FPT algorithms,
and the complexity of the fastest algorithm as a function of $k$ has led to a
long ``race'' of improving results, see 
\cite{Nied03,Chen05} and references therein.  The current best
result  is a $O^*(1.274^k)$-time polynomial-space algorithm. Another direction
of intense interest has been kernelization algorithms for \textsc{VC}, with the
current best being a kernel with (slightly less than) $2k$ vertices, see
\cite{Lampis11,Fel03,Chleb08}. Because of the importance of this problem,
numerous variations and generalizations have also been thoroughly investigated.
These include (among others): \textsc{Weighted VC} (where each vertex has a
cost), see \cite{Nied03}, \textsc{Connected VC} (where the solution is required to
be connected), see \cite{Cygan12,MolleRR08}, \textsc{Partial VC} (where the solution
size is fixed and we seek to maximize the number of covered edges), see
\cite{GuoNW07,Marx08}, and \textsc{Capacitated VC} (where each vertex has a
capacity of edges it can dominate), see \cite{DomLSV08,GuoNW07}. Of these, all
except \textsc{Partial VC} are FPT when parameterized by $k$, while all except
\textsc{Capacitated VC} are FPT when parameterized by the input graph's
treewidth $t$. \textsc{Partial VC} is known to admit an FPT approximation
scheme parameterized by $k$, see \cite{Marx08}, while \textsc{Capacitated VC} admits
a \emph{bi-criteria} FPT approximation scheme parameterized by $t$, see
\cite{Lampis14}, that is, an algorithm that returns a solution that has optimal
size, but may violate some capacity constraints by a factor $(1+\epsilon)$.

In view of the above, and the results of this paper, we observe that
\textsc{PVC} displays a different behavior than most \textsc{VC} variants, with
\textsc{Capacitated VC} being the most similar. Note though, that for
\textsc{PVC} we are able to obtain a (much simpler)
$(1+\epsilon)$-approximation for the problem, as opposed to the bi-criteria
approximation known for \textsc{Capacitated VC}. This is a consequence of a
``smoothness'' property displayed by one problem and not the other, namely,
that any solution that slightly violates the feasibility constraints of
\textsc{PVC} can be transformed into a feasible solution with almost the same
value. This property separates the two problems, motivating the further study
of \textsc{PVC}.

\section{Preliminaries}

We use standard graph theory terminology. We denote by $n$ the order of a
graph, by $N(u)$  the set of neighbors of a vertex $u$, by $d(u)=|N(u)|$ its degree, and by $\Delta$ the maximum degree of the graph. We also use standard
parameterized complexity terminology, and refer the reader to related textbook
\cite{N06} for the definitions of notions such as FPT and kernel.

In the \textsc{DPVC} problem we are given a graph $G=(V,E)$ and for each edge
$(u,v)\in E$ two positive integer values $w_{u,v}$ and $w_{v,u}$. A feasible
solution is a function that assigns to each $v\in V$ a value $p_v$ such that
for all edges we have either $p_u\ge w_{u,v}$ or $p_v\ge w_{v,u}$. If for all
edges we have $w_{u,v}=w_{v,u}$ we say that we have an instance of
\textsc{PVC}.

Both of these problems generalize \textsc{Vertex Cover}, which is the case
where $w_{u,v}=1$ for all edges $(u,v)\in E$. In fact, there are simple cases
where the problems are considerably harder.

\begin{theorem}

\label{thm:npcliques} \textsc{PVC} is NP-hard in complete graphs, even if the
weights are restricted to $\{1,2\}$. It is even APX-hard in this class of
graphs, as hard to approximate as \textsc{VC}.

\end{theorem}

\begin{proof} The reduction is from VC. Start with an instance $G=(V,E)$ of vertex cover. Put weight 2 on edges of $E$, add missing edges with weight 1. Add a new vertex $v_0$ adjacent to all other vertices with edges of weight 1. Then a vertex cover of size $k$
in the initial graph corresponds to a vertex cover of power $|V|+k$ in the final graph (put power 2 on vertices in the vertex cover, and power 1 to other vertices in $V$). The reverse is also true: in any solution all but one of the vertices must have power at least 1, and if $v_0$ is in a solution it can be replaced by another vertex. So $v_0$ is useless, and a solution of PVC takes all vertices of $V$ with power either 1 or 2; it has power $|V|+k$, where $k$ is the number of vertices with power 2. These vertices must be a vertex cover of $G$. So the $NP$-hardness follows. 

For the approximation hardness, put weights $K$ instead of 2 of edges of $G$. Then a solution of size $k$ in the initial graph corresponds to a solution of power $Kk+n$ in the final graph. Take $K=n^2$ for instance to transfer constant ratios.
\end{proof}

As a consequence of the above, \textsc{PVC} is hard on any class of graphs that
contains cliques, such as interval graphs. In the remainder we focus on classes
that do not contain all cliques, such as graphs of bounded treewidth.

We will use the standard notion of tree decomposition (for an introduction to
this notion see the survey by  \cite{BK08}).  Given a
graph $G(V,E)$ a tree decomposition of $G$ is a tree $T(I,F)$ such that every
node $i\in I$ has associated with it a set $X_i\subseteq V$, called the bag of
$i$.  In addition, the following are satisfied: $\bigcup_{i\in I} X_i = V$; for
all $(u,v)\in E$ there exists $i\in I$ such that $\{u,v\}\subseteq X_i$; and
finally for all $u\in V$ the set $\{ i\in I\ |\ u\in X_i\}$ is a connected
sub-tree of $T$. The width of a tree decomposition is defined as $\max_{i\in I}
|X_i| - 1$. The treewidth of a graph $G$ is the minimum treewidth of a tree
decomposition of $G$.

As is standard, when dealing with problems on graphs of bounded treewidth we
will assume that a ``nice'' tree decomposition of the input graph is supplied
with the input. In a nice tree decomposition the tree $T$ is a rooted binary
tree and each node $i$ of the tree is of one four special types: Leaf nodes,
which contain a single vertex; Join nodes, which have two children containing
the same set of vertices as the node itself; Introduce nodes, which have one
child and contain the same vertices as their child, plus one vertex; and Forget
nodes, which have one child and contain the same vertices as their child, minus
one vertex (see \cite{BK08} for more details).

\section{Parameterizing by treewidth}\label{sec:treewidth}

\subsection{Hardness for Treewidth}

\begin{theorem} \label{thm:twhard}

If there exists an algorithm which, given an instance $G(V,E)$ of {\sc PVC}
with treewidth $t$, computes an optimal solution in time $|V|^{o(t)}$, then the
ETH is false. This result holds even if all weights are polynomially bounded in
$|V|$.

\end{theorem}

\begin{proof}

\begin{figure}
\centering
\begin{tabular}{ccc}
\begin{minipage}{0.4\textwidth}
\includegraphics[width=\textwidth]{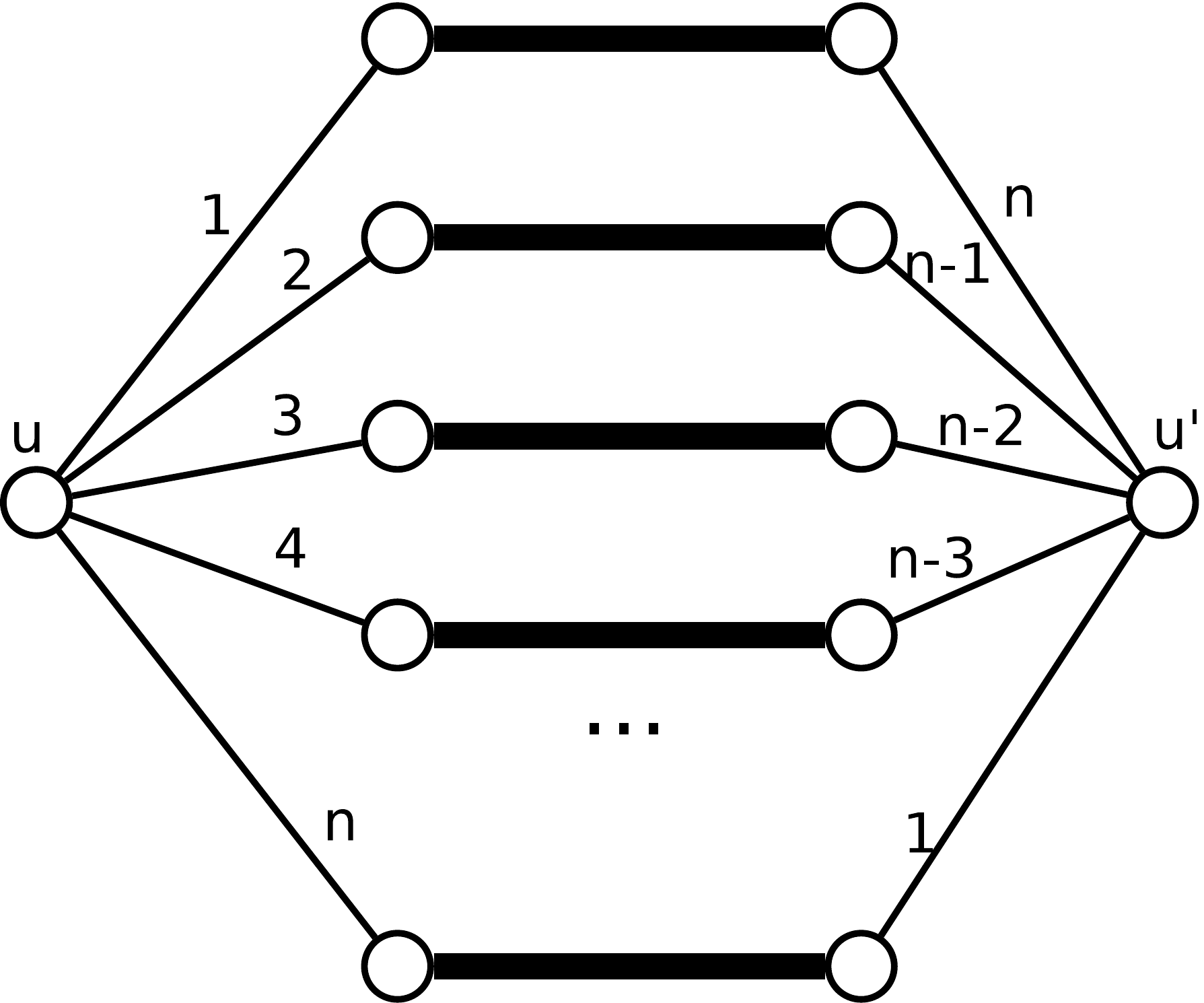}
\end{minipage}& \hspace{0.08\textwidth} &
\begin{minipage}{0.4\textwidth}
\includegraphics[width=\textwidth]{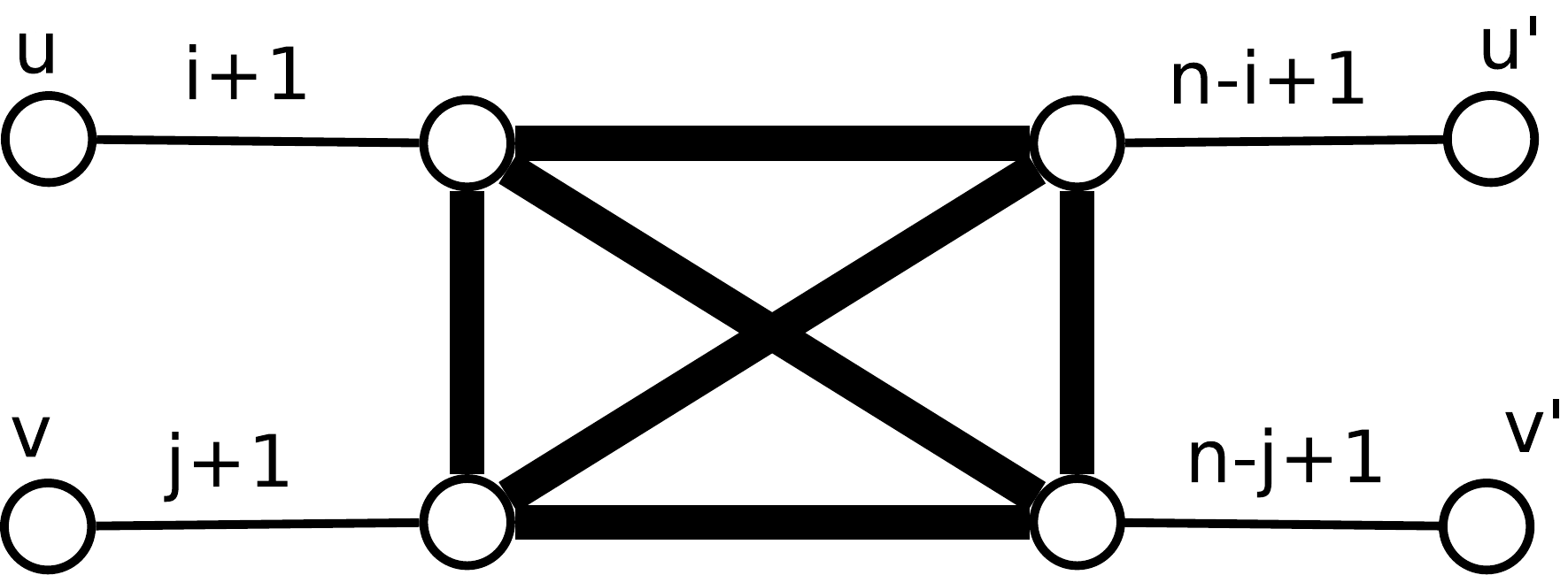}
\end{minipage}
\end{tabular}
\caption{Main gadgets of Theorem \ref{thm:twhard}. Thick lines represent weight $n$ edges.}
\label{fig:reduction}
\end{figure}

We describe a reduction from $k$-Multicolored Independent Set. In this problem
we are given a graph whose vertex set has been partitioned into $k$ cliques
$V_1,\ldots,V_k$ and we are asked if this graph contains an independent set of
size $k$. We assume without loss of generality that
$|V_1|=|V_2|=\ldots=|V_k|=n$ and that the vertices of each part are numbered
$\{1,\ldots,n\}$. It is known that if an algorithm can solve this problem in
$n^{o(k)}$ time then the ETH is false.

Our reduction relies on two main gadgets, depicted in Figure
\ref{fig:reduction}. We first describe the choice gadget, depicted on the left
side of the figure. This gadget contains two vertices $u,u'$ that will be
connected to the rest of the graph. In addition, it contains $n$ independent
edges, each of which is given weight $n$. Each edge has one of its endpoints
connected to $u$ and the other to $u'$. The weights assigned are such that no
two edges incident on $u$ have the same weight, and for each internal edge the
weight of the edges connecting it to $u,u'$ add up to $n+1$.

The first step of our construction is to take $k$ independent copies of the
choice gadget, and label the high-degree vertices $u_1,\ldots,u_k$ and
$u_1',\ldots,u_k'$. As we will see, the idea of the reduction is that the power
assigned to $u_i$ will encode the choice of vertices for the independent set in
$V_i$ in the original graph.

We now consider the second gadget of the figure (the checker), which consists
of a $K_4$, all of whose edges have weight $n$. We complete the construction as
follows: for every edge of the original graph, if its endpoints are the $i$-th
vertex of $V_c$ and the $j$-th vertex of $V_d$, we add a copy of the checker
gadget, where each of the vertices $u_c,u_c',u_d,u_d'$ is connected to a
distinct vertex of the $K_4$. The weights are $i+1, n-i+1, j+1, n-j+1$ for the
edges incident on $u_c,u_c',u_d,u_d'$ respectively.

This completes the description of the graph. We now ask if there exists a power
vertex cover with total cost at most $k(n^2+n)+3mn$, where $m$ is the number of
edges of the original graph. Observe that the treewidth of the constructed
graph is $2k+O(1)$, because deleting the vertices $u_i,u_i',
i\in\{1,\ldots,k\}$ turns the graph into a disconnected collection of $K_2$s and
$K_4$s.

First, suppose that the original graph has an independent set of size $k$. If
the independent set contains vertex $i$ from the set $V_c$, we assign the value
$i$ to $u_c$ and $n-i$ to $u_c'$. Inside each choice gadget, we consider each
edge incident on $u_c$ not yet covered, and we assign value $n$ to its other
endpoint. Similarly, we consider each edge incident on $u_c'$ not yet covered
and assign value $n$ to its other endpoint. Since all weights are distinct and
from $\{1,\ldots,n\}$, we will thus select $n-i$ vertices from the uncovered
edges incident on $u_i$ and $i$ vertices from the uncovered edges incident on
$u_i'$; thus the total value spent on each choice gadget is $n^2+n$. To see
that this assignment covers also the weight $n$ edges inside the matching,
observe that since the edges connecting each to $u,u'$ have total weight $n+1$,
at least one is not covered by $u_c$ or $u_c'$, thus one of the internal
endpoints is taken.

Let us now consider the checker gadgets. Recall that we have one such gadget
for every edge. Consider an edge between the $i$-th vertex of $V_c$ and the
$j$-th vertex of $V_d$, so that the weights are those depicted in Figure
\ref{fig:reduction}. Because we started from an independent set of $G$ we know
that for the values we have assigned at least one of the following is true:
$p_{u_c}\neq i$ or $p_{u_d}\neq j$, since these values correspond to the
indices of the vertices of the independent set. Suppose without loss of
generality that $p_{u_c}\neq i$. Therefore, $p_{u_c}>i$ or $p_{u_c}<i$. In the
first case, the edge connecting $u_c$ to the $K_4$ is already covered, so we
simply assign value $n$ to each of the three vertices of the $K_4$ not
connected to $u_c$. In the second case, we recall that we have assigned
$p_{u_c'}=n-p_{u_c}$. Therefore the edge incident on $u_c'$ is covered. Thus, in
both cases we can cover all edges of the gadget for a total cost of $3n$. Thus,
if we started from an independent set of the original graph we can construct a
power vertex cover of total cost $k(n^2+n)+3mn$.

For the other direction, suppose that a vertex cover of cost at most
$k(n^2+n)+3mn$ exists. First, observe that since the checker gadget contains a
$K_4$ of weight $n$ edges, any solution must spend at least $3n$ to cover it.
There are $m$ such gadgets, thus the solution spends at most $k(n^2+n)$ on the
remaining vertices.

Consider now the solution restricted to a choice gadget. A first observation is
that there exists an optimal solution that assigns all degree 2 vertices values
either $0$ or $n$. To see this, suppose that one such vertex has value $i$, and
suppose without loss of generality that it is a neighbor of $u$. We set its
value to $0$ and the value of $u$ to $\max\{i,p_u\}$. This is still a feasible
solution of the same or lower cost.

Suppose that the optimal solution assigns total value at most $n^2+n$ to the
vertices of a choice gadget. It cannot be using fewer than $n$ degree-two
vertices, because then one of the internal weight $n$ edges will not be covered,
thus it spends at least $n^2$ on such vertices. Furthermore, it cannot be using
$n+1$ such vertices, because then it would have to assign 0 value to $u_i,u_i'$
and some edges would not be covered. Therefore, the optimal solution uses exactly
$n$ degree-two vertices, and assigns total value at most $n$ to $u_i,u_i'$. We
now claim that the total value assigned to $u_i,u_i'$ must be exactly $n$. To
see this, suppose that $p_{u_i}+p_{u_i'}<n$. The total number of edges covered
by $u_i,u_i'$ is then strictly less than $n$. There exist therefore $n+1$ edges
incident on $u_i,u_i'$ which must be covered by other vertices. By the pigeonhole
principle, two of them must be connected to the same edge. But since we only
selected one of the two endpoints of this edge, one of the edges must be
not covered.

Because of the above we can now argue that if the optimal solution has total
cost at most $k(n^2+n)+3mn$ it must assign value exactly $3n$ to each checker
gadget and $n^2+n$ to each choice gadget. Furthermore, this can only be
achieved if $p_{u_c}+p_{u_c'}=n$ for all $c\in\{1,\ldots,k\}$. We can now see
that selecting the vertex with index $p_{u_c}$ in $V_c$ in the original graph
gives an independent set. To see this, suppose that $p_{u_c}=i$ and $p_{u_d}=j$
and suppose that an edge existed between the corresponding vertices in the
original graph. It is not hard to see that in the checker gadget for this edge
none of the vertices $u_c,u_c',u_d,u_d'$ covers its incident edge. Therefore,
it is impossible to cover every edge by spending exactly $3n$ on this gadget.
\end{proof}

\subsection{Exact and Approximation Algorithms for Treewidth}

In the previous section we showed that \textsc{PVC} is much harder than Vertex Cover,
when parameterized by treewidth. This raises the natural question of how one
may be able to work around this added complexity. Our first observation is
that, using standard techniques, it is possible to obtain FPT algorithms for
this problem by adding extra parameters. In particular, we show that the
problem is FPT if, in addition to treewidth, we consider either the maximum
weight $M$ of any edge, or the maximum degree $\Delta$ as parameters.

Since the algorithms for these two parameterizations are essentially identical,
we present a unified algorithm by considering a slightly different formulation
of the problem. In \textsc{List Directed Power Vertex Cover} (\textsc{LDPVC})
we are given the same input as in \textsc{DPVC} as well as a function $L:V\to
2^\mathbb{N}$. We say that a solution is feasible if, in addition to satisfying
the normal \textsc{DPVC} requirements, we also have for all $v\in V$ that
$p_v\in L(v)$. In other words, in this version of the problem we are also given
a function $L$ which lists the feasible power levels allowed for each vertex
$v\in V$.  In the remainder we denote by $L_{\max}$ the maximum size of any set
of feasible power levels, that is, $L_{\max} = \max_{v\in V} |L(v)|$.

\begin{theorem} \label{thm:twXPList}

There exists an algorithm which, given an instance of \textsc{LDPVC} and a tree
decomposition of width $t$ of the input graph, computes an optimal solution in
time $L_{\max}^{t+1}n^{O(1)}$.

\end{theorem}

\begin{proof}

The algorithm uses well-known DP techniques, so we only sketch the details.
Given a nice rooted tree decomposition of $G$, see \cite{BK08}, we compute a DP table
for every bag, starting from the leaves. For every bag $B$ we consider every
possible assignment of power values to the vertices of $B$; hence we consider
at most $L_{\max}^{t+1}$ such assignments, as each bag contains at most $t+1$
vertices, and each vertex has at most $L_{\max}$ allowed assignments. For each
such assignment we store the cost of the best solution that uses this
assignment for the vertices of $B$ and covers all edges incident on vertices
that only appear in bags below $B$ in the tree decomposition. It is not hard to
see how to compute such a table for Join and Forget nodes. Finally, for Insert
nodes, we consider every possible allowed value of the inserted vertex in
combination with every power assignment for the vertices of the bag. For each
such combination we check if all edges induced by the bag are covered. If that
is the case, the cost of the solution is increased by the power level of the
inserted vertex; otherwise we set it to $\infty$ to denote an infeasible
solution. \end{proof}
 
\begin{theorem} \label{thm:twXP}

There exists an algorithm which, given an instance of \textsc{DPVC} and a tree
decomposition of width $t$ of the graph, computes an optimal solution in time
$(M+1)^{t+1}n^{O(1)}$. Similarly, there exists an algorithm that performs the same
in time $(\Delta+1)^{t+1}n^{O(1)}$.

\end{theorem}

\begin{proof}

The theorem follows directly from Theorem \ref{thm:twXPList}. In particular, if
the maximum weight is $M$, it suffices to only consider solutions where all
power levels are in $\{0,\ldots,M\}$. Hence we can produce an equivalent
instance of \textsc{LDPVC} by setting $L(v)=\{0,\ldots,M\}$ for all vertices
and invoke Theorem \ref{thm:twXPList}. On the other hand, if the maximum degree
is $\Delta$, the optimal solution to \textsc{DPVC} will use for each vertex one
of $\Delta+1$ possible power values: either we give this vertex power $0$, or
we give it power equal to the weight of one of its incident edges. By
constructing the function $L$ in this way we can again invoke Theorem
\ref{thm:twXPList}.  \end{proof}

Theorem \ref{thm:twXP} indicates that the problem's hardness for treewidth is
not purely combinatorial; rather, it stems mostly from the existence of large
numbers, which force the natural DP table to grow out of proportion. Using this
intuition we are able to state the main algorithmic result of this section
which shows that, in a sense, the problem's W-hardness with respect to
treewidth is ``soft'': even if we do not add extra parameters, it is always
possible to obtain in FPT time a solution that comes arbitrarily close to the
optimal.

\begin{theorem} \label{thm:tw}

There exists an algorithm which, given an instance of \textsc{DPVC}, $G(V,E)$ and a
tree decomposition of $G$ of width $t$, for any $\epsilon>0$, produces a
$(1+\epsilon)$-approximation of the optimal in time $\left(O(\frac{\log
n}{\epsilon})\right)^tn^{O(1)}$. Therefore, \textsc{DPVC} admits an FPT approximation
scheme parameterized by treewidth.

\end{theorem}

\begin{proof}

Our strategy will be to reduce the given instance to an \textsc{LDPVC} instance
where the size of all lists will be upper bounded by $O(\frac{\log
n}{\epsilon})$, without distorting the optimal solution too much. Then, we will
invoke Theorem \ref{thm:twXPList} to obtain the stated running time.

Before we begin, observe that for \textsc{PVC}, $M$ is already a lower bound on the
value of the optimal solution. However, this is not the case for \textsc{DPVC}. We can,
however, assume that we know the largest value used in the optimal solution
(there are only $|E|$ possible choices for this value, so we can guess it). Let
$M$ be this value. Now, we know that any edge $(u,v)$ with $w_{uv}>M$ must be covered
by $v$, so we give vertex $v$ value $w_{vu}$ and adjust the graph
appropriately.  In the process we produce a graph where the largest weight is
$M$ and the optimal value is at least as high as $M$.

Now, note that in general we cannot say anything about the relation between $M$
and $n$, indeed it could be the case that the two are not even polynomially
related. To avoid this, we will first ``round down'' all weights, so that we
have $M=n^{O(1)}$. To begin with, suppose that $M>n^2$ (otherwise, all weights
are polynomial in $n$ and we are done). For every edge $(u,v)\in E$ we
calculate a new weight $w_{uv}' = \lfloor \frac{w_{uv}n^2}{M}  \rfloor$.
Observe that in the new instance we constructed the maximum weight is $M'=n^2$,
so all we need to argue is that the optimal solution did not change much. Let
$\OPT$ denote the value of the optimal solution of the original instance and
$\OPT'$ denote the optimal of the new instance. It is not hard to see that
$\OPT'\le \lfloor \frac{n^2 \OPT }{M} \rfloor \le \frac{n^2\OPT}{M} $ because
if we take an optimal solution of $G$ and divide the value of each vertex by
the same value we divided the edge weights, we will obtain a feasible solution.
Suppose now that we have a solution $\SOL'$ of the new instance such that
$\SOL'\le \rho \OPT'$, for some $\rho\ge1$. We can use it to obtain an almost
equally good solution of the original instance as follows: for every vertex of
$G$ assign it the value assigned to it by $\SOL'$, multiplied by
$\frac{M}{n^2}$, and then add to it $\frac{M}{n^2}$. The total cost of this
solution is at most $\SOL'\frac{M}{n^2}+\frac{M}{n} \le
\rho\OPT'\frac{M}{n^2}+\frac{M}{n}\le \rho\OPT + \frac{\OPT}{n} =
(\rho+\frac{1}{n})\OPT$, where we have also used the assumption that $M\ge
\OPT$. 

We therefore assume in the remainder that we are given a \textsc{DPVC} instance
with $M\le n^2$ and optimal value $\OPT\ge M$. We construct an instance of
\textsc{LDPVC} by using the same (weighted) graph and assigning to all vertices
the list $\{0\}\cup \{ (1+\epsilon)^i\ |\ 0\le i \le \lceil
2\log_{1+\epsilon}n\rceil\}$. In other words, the list of allowed power levels
consists of all integer powers of $(1+\epsilon)$ that do not go above
$(1+\epsilon)n^2$, and $0$. The size of this list is $O(\log_{1+\epsilon}n) =
O(\log n /\epsilon)$, where we have used the fact that for sufficiently small
$\epsilon$ we have $\ln(1+\epsilon)\approx \epsilon$, therefore the new
instance can be solved optimally in the claimed running time by Theorem
\ref{thm:twXPList}.

We observe that if we denote by $\OPT'$ the value of the optimal in the new
\textsc{LDPVC} instance we have $(1+\epsilon)\OPT \ge \OPT'\ge \OPT$: the
second inequality is trivial, while the first follows from the fact that any
optimal solution to the \textsc{DPVC} instance can be transformed into a
feasible solution of the \textsc{LDPVC} instance by replacing every power level
$p_v$ with the value $(1+\epsilon)^{\lceil \log_{1+\epsilon}p_v\rceil}$, or in
other words, with the smallest integer power of $(1+\epsilon)$ which is higher
than $p_v$ (this is because we can assume that the optimal solution of the
\textsc{DPVC} never uses a power level above $M$, and $M\le n^2$).  Hence, the
solution produced by the algorithm of Theorem \ref{thm:twXPList} can be
transformed into a solution of the original instance with almost the same
value.

Finally, we note that the running time in the theorem implies an FPT running
time (parameterized by $t+1/\epsilon$) for the approximation scheme via
standard arguments.  First, if $t+1/\epsilon<\log n/\log\log n$ then the
running time of the algorithm is actually polynomial, hence FPT, because we
have $\log n/\epsilon < \log^2n$, $t<\log n/\log\log n $, therefore $(\log
n/\epsilon)^t < (\log^2n)^{\log n/\log\log n} = n^{2}$.  If on the other hand
$t+1/\epsilon>\log n/\log \log n$ then we have $\log n < (t+1/\epsilon)^2$
hence the running time is at most $\left(t/\epsilon\right)^{O(t)} n^{O(1)}$,
which is also FPT.  
\end{proof}

\section{Parameterizing by total power}\label{sec:standard}

We focus in this section on the standard parameterization: given an edge-weighted graph $G$ and an integer $P$ (the parameter), we want to determine if there exists a solution of total power at most $P$. We first focus on {\sc PVC} and show that it is solvable within time $O^*(1.274^P)$, thus reaching the same bound as {\sc VC} (when parameterized by the solution value). We then tackle {\sc DPVC} where a more involved analysis is needed, and we reach time $O^*(1.325^P)$.


\subsection{PVC}

The algorithm for {\sc PVC} is based on the following simple property.

\begin{property}\label{prop1} Consider an edge $e=(u,v)$ of maximum weight. Then, in any optimal solution either $p_u=w_e$ or $p_v=w_e$.
\end{property}

This property can be turned into a branching rule: considering an edge $e=(u,v)$ of maximum weight, then either set $p_u=w_e$ (remove $u$ and incident edges), or set $p_v=w_e$ (remove $v$ and incident edges). This already shows that the problem is FPT, leading to an algorithm in time $O^*(2^P)$.
To improve this and get the claimed bound, we also use the following reduction rule.
\begin{enumerate}
\item[(RR1)] Suppose that there is an edge $(u,v)$ with $w_{u,v}=M$, and the maximum weight of other edges incident to $u$ or $v$ is $B\leq M-1$. Then set $w_{u,v}=B$, and do $P\leftarrow P-(M-B)$.
\end{enumerate}
\begin{property}\label{prop2}
(RR1) is correct.
\end{property}

\begin{proof} Take a solution with the initial weights. Then $p_u=M$ or
$p_v=M$, and we get a solution on the modified instance by reducing the power
of $u$ or $v$ by $M-B$. Conversely, if we have a solution with the modified
weights, either $p_u=B$ or $p_v=B$, and we get a solution with the initial
weights by adding power value $M-B$ to $u$ or to $v$. \end{proof}

Now, consider the following branching algorithm.

\fbox{
\begin{minipage}{0.90\textwidth}

\begin{tabular}{ll}
& {\bf Algorithm 1} \\
\hline
\hline

STEP 1:& While (RR1) is applicable, apply it;\\

STEP 2:& If $P<0$ return NO;\\
	& If the graph has no edge return YES;\\

STEP 3:& If the maximum weight of edges is 1:\\
& \hspace{0.5cm}Apply an algorithm in time $O^*(1.274^k)$ for {\sc VC};\\

STEP 4:& Take two adjacent edges $e=(u,v)$ and $f=(u,v')$ of\\
& maximum weight. Branch as follows:\\
    & - either set $p_u=w_e$ (remove $u$, set $P\leftarrow P-w_e$);\\
    & - or set $p_v=p_{v'}=w_e$ (remove $v$ and $v'$, set $P\leftarrow P-2w_e$).\\
\end{tabular}
\end{minipage}
}

\begin{theorem} \label{thm:alg1}
Algorithm 1 solves {\sc PVC} in time $O^*(1.274^P)$.
\end{theorem}

\begin{proof}
When (RR1) is no longer applicable, any edge of maximum weight is adjacent to another edge of the same (maximum) weight. So in Step 4 we can always find such a pair of edges $e,f$. Then the validity of the algorithm follows from Properties~\ref{prop1} and~\ref{prop2}. For the running time, the only step to look at is Step 4. When branching, the maximum weight is $w_e=w_f\geq 2$. In the first branch $P$ is reduced by at least 2, in the second branch by at least 4, and this gives a branching factor $(-2,-4)$ inducing a complexity smaller than the claimed bound.  
\end{proof}

\subsection{DPVC}

For {\sc DPVC}, the previous simple approach does not work. Indeed, there might be no pair of vertices $(u,v)$ such that both $w_{u,v}=\max\{w_{u,z}:z\in N(u)\}$ and $w_{v,u}=\max\{w_{v,z}:z\in N(v)\}$.
If we branch on a pair $(u,v)$, the only thing that we know is that either $p_u\geq w_{u,v}$, or $p_v\geq w_{v,u}$. Setting a constraint $p_u\geq w$ corresponds to the following operation $Adjust(u,w)$.

\begin{definition}
$Adjust(u,w)$ consists of decreasing each weight $w_{u,z}$ for $z\in N(u)$ by $w$, and decreasing $P$ by $w$.
\end{definition}
Of course, if a weight $w_{u,z}$ becomes 0 or negative, then the edge $(u,z)$ is removed from the graph (it is covered by $u$).

In our algorithm, a typical branching will be to take an edge $(u,v)$, and either to apply $Adjust(u,w_{u,v})$ (and solve recursively), or to apply $Adjust(v,w_{v,u})$ (and solve recursively). Another possibility is to set the power of a vertex $u$ to a certain power $w$. In this case we must use $v$ to cover $(v,u)$ if $w_{u,v}>w$. Formally:
\begin{definition}
$Set(u,w)$ consists of (1) setting $p_u=w$, removing $u$ (and incident edges), (2) decreasing $P$ by $w$, and (3) applying $Adjust(v,w_{v,u})$ for all $(v,u)$ such that $w_{u,v}>w$.
\end{definition}
Using this, it is already easy to show that the problem is solvable in  FPT-time $O^*(1.619^P)$. To reach the claimed bound of $O^*(1.325^P)$ we need some more ingredients.
Let $M(u)=\max\{w_{u,z}:z\in N(u)\}$ and $P(u)=\sum_{z\in N(u)}w_{z,u}$. We first define two reduction rules and one branching rule.

\begin{enumerate}
\item[(RR2)] If there exists $u$ such that $P(u)\leq M(u)$, do $Adjust(v,w_{v,u})$ where $v$ is such that $w_{u,v}=M(u)$.
\end{enumerate}

\begin{enumerate}
\item[(RR3)] If there is $(u,v)$ with $w_{u,v}=w_{v,u}=2$, $w_{u,z}=1$ for all $z\in N(u)\setminus v$, and $w_{v,z}=1$ for all $z\in N(v)\setminus u$, then set $w_{u,v}=w_{v,u}=1$, and do $P\leftarrow P-1$.
\end{enumerate}

\begin{enumerate}
\item[(BR1)] If there exists $u$ with $P(u)\geq 5$, branch as follows: either $Set(u,0)$, or $Adjust(u,1)$.
\end{enumerate}
\begin{property}\label{prop3}
(RR2) and (RR3) are correct.
(BR1) has a branching factor (at worst) $(-1,-5)$.
\end{property}

\begin{proof}
For (RR2):  It is never interesting to put $p_u=M(u)$, since we can cover all edges incident to $u$ by putting power at most $P(u)$ on the neighbors of $u$. So, we can assume that the power of $u$ is smaller than $M(u)$, meaning that $v$ covers $(v,u)$.

For (RR3): Take a solution with the initial weights. Then $p_u=2$ or $p_v=2$, and we get a solution on the modified instance by reducing the power of $u$ (or $v$) by 1. Conversely, if we have a solution with the modified weights, either $p_u=1$ or $p_v=1$ and we get a solution with the initial weights by adding power 1 to $u$ or $v$.

For (BR1): When setting $p_u=0$, we need to cover all edges incident to $u$ using the other extremity. Since $P(u)\geq 5$, $P$ reduces by at least 5 in this branch (and at least 1 in the other branch).
\end{proof}

Now, before giving the whole algorithm, we detail the case where the maximum weight is 2, where a careful case analysis is needed.

\begin{lemma}\label{lemmaweight2}
Let us consider an instance where (1) (RR2) and (RR3) have been exhaustively applied, and (2) the maximum edge-weight is $w_{u,v}=2$. Then there is a branching algorithm with branching factor (at worst) $(-1,-5)$ or $(-2,-3)$.
\end{lemma}

\begin{proof}
Let us consider an instance where (1) (RR2) and (RR3) have been exhaustively applied, and (2) the maximum edge-weight is $w_{u,v}=2$.

We show by a case analysis that we can reach branching factor at worst $(-1,-5)$ or $(-2,-3)$.

If $P(u)\geq 5$, then we apply $(BR1)$. Otherwise, $P(u)\leq 4$. Since $P(u)>M(u)$ (otherwise (RR2) applies), $P(u)\in \{3,4\}$. In particular, the degree of $u$ is $d(u)\in \{2,3,4\}$.

Note first that if for all $z\in N(u)$ we have $w_{u,z}=2$, then we either $Set(u,0)$ (and we fix $P(u)\geq 3$ on neighbors of $u$), or $Set(u,2)$, meaning that $P$ reduces by 2. This gives a branching factor at worst $(-3,-2)$. Thus, we can assume in the sequel that there exists $z\in N(u)$ such that $w_{u,z}=1$.

We need to consider several cases, taking a vertex $u$ (incident to an edge of weight 2) of maximum degree:
\begin{enumerate}
\item If $d(u)=4$. Then since $P(u)\leq 4$ all edges $(z,u)$ have weight $w_{z,u}=1$.
\begin{enumerate}
    \item If there is only one $z$ with $w_{u,z}=1$: it is never interesting to put power 1 on $u$. So we do either $Set(u,0)$ (and we fix 4), or $Set(u,2)$ and we fix 2.
    \item If there are three neighbors $z$ with $w_{u,z}=1$: it is never interesting to set $p_u=2$ (instead set $p_u=1$ and give the extra power 1 to $v$ to cover $(v,u)$), so $(u,v)$ is covered by $v$ and without branching we do $Adjust(v,1)$.
    \item If $w_{u,v}=w_{u,z}=2$, $w_{u,s}=w_{u,t}=1$. If both $v$ and $z$ have degree 1, then it is never interesting to take them, so we do $Set(u,2)$ without branching. Otherwise, suppose that $z$ has degree at least 2. Then we branch on $z$: we do either $Set(z,0)$ (and we fix at least 3: 2 for $Adjust(u,2)$ and 1 for the other neighbor), or $Set(z,1)$. But in this latter case it is never interesting to set $p_u=2$ (set $p_u=1$ and put extra power 1 on $v$ instead), so we also do $Adjust(v,1)$. The branching factor is then at worst $(-3,-2)$.
\end{enumerate}
\item If $d(u)=3$, with $N(u)=\{v,z,t\}$:
\begin{enumerate}
    \item Suppose first that $w_{u,z}=2$ and $w_{u,t}=1$.

    If $w_{t,u}=1$ it is never interesting to set $p_u=1$ (set $p_u=0$ and put extra weight on $t$ instead), so we do either $Set(u,0)$ (and fix at least $P(u)\geq 3$ in the set $N(u)$ of neighbors of $u$), or $Set(u,2)$. We get a branching factor $(-3,-2)$.

    Otherwise, $w_{t,u}=2$. Since $P(u)\leq 4$, $w_{v,u}=w_{z,u}=1$. Then if $d(v)=d(z)=1$, we do not need to branch, we do $Set(u,2)$. Otherwise, for instance $d(z)\geq 2$. We branch on $z$: we do either $Set(z,0)$ and we fix at least 3 in $N(z)$, or $Adjust(z,1)$ but in this case it is never interesting to set $p_u=2$ (set $p_u=1$, and put the extra weight 1 on $v$ instead), so we do $Adjust(v,1)$. We get a branching factor $(-3,-2)$.
    \item Now suppose that $w_{u,z}=w_{u,t}=1$.

    If $w_{v,u}=1$ it is never interesting to set $p_u=2$, so we do $Adjust(v,1)$ without branching. The only remaining case is $w_{v,u}=2$.

    If $v$ has another neighbor $s$ such that $w_{v,s}=2$: we branch on $v$, by doing either $Set(v,2)$, or by considering $w(v)<2$ hence $Adjust(u,2)$ and $Adjust(s,w_{s,v})$. We get a branching factor $(-2,-3)$.

    Otherwise, we have $(u,v)$ with $w_{u,v}=w_{v,u}=2$, $w_{u,z}=1$ for all $z\in N(u)\setminus \{v\}$ and $w_{v,z}=1$ for all $z\in N(v)\setminus \{u\}$. But this cannot occur thanks to Rule (RR3).
       
\end{enumerate}
\item If $d(u)=2$, $N(u)=\{v,z\}$. Then $w_{u,z}=1$ (since as mentioned before there exists $z\in N(u)$ such that $w_{u,z}=1$).
\begin{enumerate}
\item If $w_{z,u}=1$ it is never interesting to set $p_u=1$. As previously, we do either $Set(u,0)$ (and fix $P(u)\geq 3$ in $N(u)$), or $Set(u,2)$ and reduce $P$ by two. So again we get a branching factor $(-3,-2)$.
\item If $w_{v,u}=1$, it is never interesting to set $p_u=2$, so we do $Adjust(v,1)$.
\item Otherwise, $w_{v,u}=w_{z,u}=2$. Let us look at $z$: it has degree 2 with $w_{z,u}=2$ and $w_{u,z}=1$ so this is a previous case with $u=z$.
\end{enumerate}
\end{enumerate}
\end{proof}

We are now ready to describe the main algorithm. $N^2(u)$ denotes the set of vertices at distance 2 from $u$.\\

\fbox{
\begin{minipage}{0.93\textwidth}

\begin{tabular}{ll}
& {\bf Algorithm 2} \\
\hline
\hline

STEP 1:& While (RR2) or (RR3) is applicable, apply it;\\

STEP 2:& If $P<0$ return NO;\\
	& If the graph has no edge return YES;\\

STEP 3:& If there is a vertex $u$ with $P(u)\geq 5$: apply (BR1);\\

STEP 4:& If there exists $(u,v)$ with either $w_{u,v}+w_{v,u}\geq 6$, or  \\
& $w_{u,v}=2$ and $w_{v,u}=3$: branch by either $Adjust(u,w_{u,v})$,\\
& or $Adjust(v,w_{v,u})$;\\

STEP 5:& If there exists $(u,v)$ of weight $w_{u,v}=3$:\\
& - if $N^2(u)=\{t\}$ and  $w_{t,z}=1$ for all $z\in N(u)$:\\
    & \hspace{0.5cm} - either $Adjust(t,1)$, and solve the problem on $G[N[u]]$; \\
    & \hspace{0.5cm} - or $Set(t,0)$; \\
    & - otherwise:\\
    & \hspace{0.5cm} - either $Adjust(v,1)$;\\
    & \hspace{0.5cm} - or $Set(u,3)$, and $Set(z,0)$ for all $z\in N(u)$;  \\

STEP 6:& If the maximum weight is at most 2, then branch as in Lemma~\ref{lemmaweight2}.
\end{tabular}
\end{minipage}
}

\vspace{0.3cm}

In Step 5, when doing $Adjust(t,1)$, as we will see in the proof $N[u]$ becomes separated from the rest of the graph, ie. a connected component with at most 5 vertices, so we can find an optimal solution on it in constant time.

\begin{theorem} \label{thm:alg2}
    Algorithm 2 solves {\sc DPVC} in time $O^*(1.325^P)$.
\end{theorem}

\begin{proof}
Note that $1.325$ corresponds to branching factors $(-1,-5)$ and $(-2,-3)$.

We have already seen that (RR2) and (RR3) are sound, and that (BR1) gives a branching factor $(-1,-5)$.

For Step 4, if $w_{u,v}+w_{v,u}\geq 6$ this gives in the worst case a branching factor $(-1,-5)$. If $w_{u,v}=2$ and $w_{v,u}=3$ the branching factor is $(-2,-3)$.

At this point (after Step 4) there cannot remain an edge $(u,v)$ with weight $w_{u,v}\geq 5$, since Step 4 would have been applied. If there is an edge $(u,v)$ with $w_{u,v}=4$, then either $P(u)\geq 5$ (which is impossible since (BR1) would have been applied), or $P(u)\leq 4$ (which is impossible since (RR2) would have been applied). So, the maximum edge weight after step 4 is at most 3.

Thanks to Lemma~\ref{lemmaweight2}, what remains to do is to focus on Step 5, with $w_{u,v}=3=M(u)$.

First, note that then $P(u)=4$ (otherwise (RR2) or (BR1) would have been applied), and $w_{v,u}=1$ (otherwise Step 4 would have been applied).

We consider two cases.
\begin{itemize}
        \item If $N^2(u)=\{t\}$ and $w_{t,z}=1$ for all $z\in N(u)$. As explained in the algorithm, we branch on $t$: either $p_t\geq 1$, or $p_t=0$. If $p_t\geq 1$, all edges $(t,z),z\in N(u)$ are covered, so $N[u]$ is now disconnected from the rest of the graph. We can find (in constant time) the power to put on these vertices in order to optimally color the incident edges, as mentioned in the algorithm. Since we need power at least 2 for this ($u$ has at least two neighbors since $P(u)=4$ and $w_{v,u}=1$), $P$ globally reduces by at least 3. If $p_t=0$, $Set(t,0)$ induces to perform $Adjust(z,w_{z,t})$ for all $z\in N(t)$. $t$ has at least 2 neighors since $N(u)\subseteq N(t)$, and $P$ reduces by at least 2. Then, $P$ reduces by at least 3 in one branch, by at least 2 in the other, leading to a branching factor $(-3,-2)$.
    \item If $|N^2(u)|\geq 2$, or if $N^2(u)=\{t\}$ with at least one $z \in N(u)$ with $w_{t,z}=2$: setting $p_u=3$ is only interesting if all neighbors of $u$ receive weight 0 - otherwise distributing the power 3 of $u$ on neighbors of $u$ to cover all edges $(v,u)$ is always at least as good. Then in this case we have to cover edges between $N(u)$ and $N^2(u)$, so a power at least 2. Then either we have  $p_u<3$ and in this case $p_v\geq 1$, or we set $p_u=3$, $p_z=0$ for neighbors $z$ of $u$, and we fix at least  2 in $N^2(u)$. In the first branch $P$ reduces by at least 1, in the other by at least 5, leading to a branching factor $(-1,-5)$.
\end{itemize}
\end{proof}

\section{Parameterizing by the number $k$ of vertices}


We now consider as parameter the number $k$ of vertices that will
receive a positive value in the optimal solution. Note that by definition $k\le
P$; therefore, we expect any FPT algorithm with respect to $k$ to have the same or worse
performance than the best algorithm for parameter $P$.

\begin{theorem}\label{theop}
{\sc PVC} is solvable in time $O^*(1.619^k)$. {\sc DPVC} is solvable in time $O^*(k^k)$.
\end{theorem}

\begin{proof}
For {\sc PVC}, Algorithm 1 can be easily adapted to deal with parameterization by $k$: when branching, $k$ is reduced by 1 in the first branch, and by 2 in the second branch, so we have a branching factor $(-1,-2)$ in the worst case. 

The case of {\sc DPVC} is a bit more involved, but we can show that it is FPT as well. We will perform procedures $Adjust$ and $Set$ (without the modifications on $P$ of course), but we have to pay attention to the way we count the number of vertices in the solution (vertices with strictly positive power). Indeed, if we perform several $Adjust(u,w)$ on the same vertex $u$, we have to count it (decrease $k$) only once. To do this, in the algorithm we decrease $k$ by one when a vertex (with positive weight) is removed from the graph. Then, when performing $Adjust(u,w)$, we mark the vertex $u$ (but do not decrease $k$ by one yet) - initially no vertex is marked. When a marked vertex becomes of degree 0, then we remove it from the graph and decrease $k$ by one. Also, when performing an operation $Set(u,w)$ ($u$ is removed), so we decrease $k$ by 1  if $w>0$ or if $u$ is marked. 

Now, suppose first that there is a vertex $u$ of degree at least $k+1$. To get a solution with at most $k$ vertices, necessarily $u$ must receive a positive power. More precisely, it must cover all but at most $k$ edges incident to it. Then in any solution with at most $k$ vertices, $p_u\geq w_{k+1}$ if we order the weights $w_{u,v},v\in N(u)$, in non increasing order $w_1\geq w_2\geq \dots \geq w_d$. So in this case, we can safely apply $Adjust(u,w_{k+1})$. In the remaining instance, $u$ has degree at most $k$.

We apply this reduction while possible, and get afterwards an instance where each vertex has degree at most $k$. Then:
\begin{itemize}
\item Take an edge $(u,v)$ such that $w_{u,v}=\max\{w_{u,z}:z\in N(u)\}$. Let $\{w_1,w_2,\dots,w_s\}$ be the set of weights $w_{v,z},z\in N(v)$ which are greater than or equal to $w_{v,u}$.
\item Perform a branching with $s+1$ branches as follows: either $Set(u,w_{u,v})$ (and recurse), or $Set(v,w_i)$ for each $i=1,\dots,s$ (and recurse).
\end{itemize}

Choosing $u$ to cover the edge $(u,v)$ corresponds to the first branch. Otherwise, we have to use $v$ to cover this edge. We consider all the possible weights that could be assigned to $v$ in an (optimal) solution. In each branch at least one vertex with positive power is removed ($u$ in the first branch, $v$ in the $s$ others) so $k$ reduces by at least one. Since $s\leq k$, we have a tree of arity at most $k+1$ and depth at most $k$.
\end{proof}

Following Theorem~\ref{theop}, a natural question is whether {\sc DPVC} is solvable in single exponential time with respect to $k$ or not. This does not seem obvious. In particular, it is not clear whether {\sc DPVC} is solvable in single exponential time {\it with respect to the number of vertices $n$}, since the simple brute-force algorithm which guesses the value of each vertex needs $n^{O(n)}$ time.

Interestingly, though we are not able to resolve these questions, we can show that they are actually equivalent.

\begin{theorem} \label{thm:dpvcexp}If there exists an $O^*(\gamma^n)$-time algorithm for \textsc{DPVC}, for some constant $\gamma>1$, then
there exists an $O^*((4\gamma)^k)$-time algorithm for \textsc{DPVC}.

\end{theorem}

\begin{proof}

We describe a branching algorithm, which eventually needs to solve an instance
of \textsc{DPVC} on $k$ vertices. We maintain three sets $C_1,C_2,I\subseteq V$.
Initially, $I=V$ and $C_1,C_2=\emptyset$. The informal meaning is that we want
to branch so that $I$ eventually becomes an independent set and $C_1\cup C_2$ a
vertex cover of $G$. In addition, we want to maintain the property that there
are no edges between $C_2$ and $I$ (which is initially vacuously true).

As long as $|C_1\cup C_2|\le k$ and there are still edges incident on $I$ we
repeat the following rules:

\begin{enumerate}

\item If there is a vertex $u\in C_1$ such that there are no edges from $u$ to
$I$ we set $C_2:=C_2\cup \{u\}$ and $C_1:=C_1\setminus\{u\}$.

\item If there is an edge $(u,v)$ induced by $I$ we branch on which of the two
vertices covers it in the optimal solution. In one branch, we perform
$Adjust(u,w_{u,v})$, and set $I:=I\setminus\{u\}$ and $C_1:=C_1\cup\{u\}$. The
other branch is symmetric for $v$.

\item If $I$ is an independent set, we select a vertex $u\in C_1$. Because of
Step 1, we know that there are some edges connecting $u$ to $I$. Let $w_{u,v}$
be the maximum weight among edges connecting $u$ to $I$. We branch between two
choices: either in the optimal $p_u\ge w_{u,v}$ or $p_v\ge w_{v,u}$. In the first
case we $Adjust(u,w_{u,v})$ (observe that this removes all edges connecting $u$
to $I$, therefore we can immediately apply Rule 1). In the second case, we
$Adjust(v,w_{v,u})$ and set $|I|:=I\setminus\{v\}$ and $C_1:=C_1\cup\{v\}$.

\end{enumerate}

It is not hard to see that the branching described above cannot produce more
than $4^k$ different outcomes. To see this, observe that any branching
performed either increases $|C_1\cup C_2|$ or $|C_2|$, while never decreasing
either quantity.  Therefore, the quantity $|C_1\cup C_2|+|C_2|$ increases in
each branching step. However, this quantity has maximum value $2k$.

The branching algorithm above will stop either when $I$ has no incident edges,
or when $|C_1\cup C_2|=k$. In the latter case, we have selected that $C_1\cup C_2$
is the set of vertices that take positive values in our solution, so we can
perform $Adjust(u,w_{u,v})$ for any $u\in C_1$ and $v\in I$, until no edges are
incident on $I$. We can now delete all the (isolated) vertices of $I$, and we
are left with a $k$-vertex \textsc{DPVC} instance on $C_1\cup C_2$, on which we use the
assumed exponential-time algorithm. 
\end{proof}

\section{Kernelization and linear programming}

Moving to the subject of kernels, we first notice that the same technique as for {\sc VC} gives a quadratic (vertex) kernel for {\sc DPVC} when the parameter is $k$ (and therefore also when the parameter is $P$):

\begin{theorem} \label{thm:kernel}

There exists a kernelization algorithm for \textsc{DPVC} that produces a kernel with $O(k^2)$ vertices.

\end{theorem}

\begin{proof}

Consider the following rules:

\begin{itemize}
\item If there exists a vertex $u$ of degree at least $k+1$, apply the procedure $Adjust(u,w_{k+1})$ as previously (and mark the vertex).
\item If there exists a marked vertex of degree 0, remove it and do $k\leftarrow k-1$. If there is an unmarked vertex of degree 0, remove it.
\end{itemize}
Let $G'$ be the graph obtained after this procedure. If $n'>k(k+1)$ return NO.
Else, $n'\leq (k+1)^2$ which is a kernel.

Suppose that the answer is YES. Then there exists a vertex cover of size at most $k$, and each vertex has at most $k$ neighbors, and there is no isolated vertex, so $n'\leq k(k+1)$. Hence, this gives a kernel of size $O(k^2)$ for {\sc DPVC} when parameterized by $k$ (so a kernel of size $O(P^2)$ when parameterized by $P$). 
\end{proof}

We observe that the above theorem gives a \emph{bi-kernel} also for \textsc{PVC}. We leave it as an open question whether a pure quadratic kernel exists for \textsc{PVC}.

Let us now consider the question whether the kernel of Theorem \ref{thm:kernel}
could be improved to linear.  A way to reach a linear kernel for {\sc VC} is by
means of linear programming. We consider this approach now and show that it
seems to fail for the generalization we consider here.
Let us consider the following ILP formulation for {\sc DPVC}, where we have one variable per vertex ($x_i$ is the power of $u_i$), and one variable $x_{i,j}$ ($i<j$) per edge $(u_i,u_j)$. $x_{i,j}=1$ (resp. 0) means that $u_i$ (resp $u_j$) covers the edge.

\begin{eqnarray*}
\left\{ \begin{array}{ll}
 \min \ \sum_{i=1}^n x_i\\
s.t. \left | \begin{array}{llll}
x_i & \geq & w_{i,j}x_{i,j}, &\forall (u_i,u_j)\in E, i<j\\
x_j & \geq & w_{j,i}(1-x_{i,j}), &\forall (u_i,u_j)\in E, i<j\\
x_{i,j}& \in &\{0,1\}, &\forall (u_i,u_j)\in E, i<j\\
x_i & \geq & 0, &\forall i=1,\cdots,n
\end{array}\right.
\end{array}\right.
\end{eqnarray*}

Can we use the relaxation of this ILP to get a linear kernel? Let us focus on {\sc PVC}, where the relaxation can be written in an equivalent simpler form\footnote{A solution of the relaxation of the former is clearly a solution of the latter. Conversely, if $x_i+x_j\geq w_{i,j}$, set $x_{i,j}=x_i/w_{i,j}$ to get a solution of the former.}:

\begin{eqnarray*}
\left\{ \begin{array}{ll}
 \min \ \sum_{i=1}^n x_i\\
s.t. \left | \begin{array}{rlll}
x_i+x_j & \geq & w_{i,j}, & \forall (u_i,u_j)\in E, i<j\\
x_i & \geq & 0, &\forall i=1,\cdots,n
\end{array}\right.
\end{array}\right.
\end{eqnarray*}

Let us call $RPVC$ this LP. We can show that, similarly as for {\sc VC}, $RPVC$
has the semi-integrality property:  in an optimal (extremal) solution $x^*$,
$2x^*_i\in \mathbb{N}$ for all $i$. However, we {\it cannot} remove vertices
receiving value 0, as in the case of VC. Indeed, there  does exist vertices
that receive weight 0 in the above relaxation which are in {\it any} optimal
(integer) solution.  To see this, consider two edges $(u_1,v_1)$ and
$(u_2,v_2)$ both with weight 2, and a vertex $v$ adjacent to all 4 previous
vertices with edges of weight 1. Then, there is only one optimal fractional
solution, with $p_{u_1}=p_{v_1}=p_{u_2}=p_{v_2}=1$, and $p_v=0$. But any
(integer) solution has value 5 and gives power 2 to $u_1$ or $v_1$, to $u_2$ or
$v_2$, and weight 1 to $v$.  The difficulty is actually deeper, since we have
the following.

\begin{theorem} \label{thm:NPCILP}
The following problem is NP-hard: given an instance of {\sc PVC}, an optimal (extremal) solution $x^*$ of $RPVC$ and $i$ such that $x^*_i=0$, does there exists an optimal (integer) solution of \textsc{PVC} not containing $v_i$?
\end{theorem}

\begin{proof}
Take a graph $G=(V,E)$ which is an instance of VC. We build the following instance $G'=(V',E')$ of {\sc PVC}:
\begin{itemize}
    \item There is one vertex for each vertex in $V$, and two vertices $v'_e$ and $v_e''$ for each edge $e$ in $E$
    \item $v'_e$ and $v''_e$ are linked with an edge of weight 2;
    \item If $e=(u,v)$ is an edge of $G$, then $u$ is adjacent to $v'_e$ and $v$ is adjacent to $v_e''$ in $G'$, both edges receiving weight 1.
\end{itemize}
As previously, an optimal solution in the relaxation {\sc RPVC} is to give power 1 to each ``edge-vertex'' ($v'_e,v_e''$), and 0 to vertices corresponding to vertices in $V$.

Then, consider an optimal solution $C'$ of {\sc PVC}: for each $e\in G$, $p_{v'_e}=2$ or $p_{v_e''}=2$ - to cover the edge. Moreover, if say $p_{v'_e}=2$, then we can fix $p_{v_e''}=0$: indeed, there is only one more edge incident to $v_e''$ to cover, so
we shall put power 1 to the other neighbor ($v$) of $v_e''$. So $C'=W\cup V'$, where $W$ is made of $|E|$ ``edge-vertices'', and $V'$ is made of vertices of $V$. To be feasible, $V'$ must be a vertex cover of $G$. Conversely, if $V'$ is a vertex cover of $G$, then we can easily add a set $W$ of $|E|$ ``edge-vertices'' to get a feasible solution of {\sc PVC}.

Then, a vertex in $V$ is in an optimal VC of $G$ iff it is in an optimal solution of \textsc{PVC}. Since it receives weight 0 in the (unique) optimal solution of the relaxation, the result follows. 
\end{proof}

\acknowledgements
\label{sec:ack}
We thank an anonymous referee for his/her pertinent comments and suggestions.

\bibliographystyle{abbrvnat}
\bibliography{biblio}

\begin{thebibliography}{16}
\providecommand{\natexlab}[1]{#1}
\providecommand{\url}[1]{\texttt{#1}}
\expandafter\ifx\csname urlstyle\endcsname\relax
  \providecommand{\doi}[1]{doi: #1}\else
  \providecommand{\doi}{doi: \begingroup \urlstyle{rm}\Url}\fi

\bibitem[Angel et~al.(2015)Angel, Bampis, Chau, and Kononov]{Angel2015}
E.~Angel, E.~Bampis, V.~Chau, and A.~Kononov.
\newblock Min-power covering problems.
\newblock In \emph{ISAAC 2015}, volume 9472 of \emph{LNCS}, pages 367--377.
  Springer, 2015.

\bibitem[Angel et~al.(2016)Angel, Bampis, Escoffier, and
  Lampis]{DBLP:conf/wg/AngelBEL16}
E.~Angel, E.~Bampis, B.~Escoffier, and M.~Lampis.
\newblock Parameterized power vertex cover.
\newblock In \emph{{WG} 2016}, volume 9941 of \emph{LNCS}, pages 97--108.
  Springer, 2016.

\bibitem[Bodlaender and Koster(2008)]{BK08}
H.~L. Bodlaender and A.~M. C.~A. Koster.
\newblock Combinatorial optimization on graphs of bounded treewidth.
\newblock \emph{Comput. J.}, 51\penalty0 (3):\penalty0 255--269, 2008.

\bibitem[Chen et~al.(2010)Chen, Kanj, and Xia]{Chen05}
J.~Chen, I.~A. Kanj, and G.~Xia.
\newblock Improved upper bounds for vertex cover.
\newblock \emph{Theoretical Computer Science}, 411\penalty0 (40–42):\penalty0
  3736 -- 3756, 2010.

\bibitem[Chleb{\'{\i}}k and Chleb{\'{\i}}kov{\'{a}}(2008)]{Chleb08}
M.~Chleb{\'{\i}}k and J.~Chleb{\'{\i}}kov{\'{a}}.
\newblock Crown reductions for the minimum weighted vertex cover problem.
\newblock \emph{Discrete Applied Mathematics}, 156\penalty0 (3):\penalty0
  292--312, 2008.

\bibitem[Cygan(2012)]{Cygan12}
M.~Cygan.
\newblock Deterministic parameterized connected vertex cover.
\newblock In \emph{{SWAT 2012}}, volume 7357 of \emph{LNCS}, pages 95--106.
  Springer, 2012.

\bibitem[Dom et~al.(2008)Dom, Lokshtanov, Saurabh, and Villanger]{DomLSV08}
M.~Dom, D.~Lokshtanov, S.~Saurabh, and Y.~Villanger.
\newblock Capacitated domination and covering: {A} parameterized perspective.
\newblock In \emph{{IWPEC 2008}}, volume 5018 of \emph{LNCS}, pages 78--90.
  Springer, 2008.

\bibitem[Fellows(2003)]{Fel03}
M.~R. Fellows.
\newblock Blow-ups, win/win's, and crown rules: Some new directions in {FPT}.
\newblock In \emph{{WG 2003}}, volume 2880 of \emph{LNCS}, pages 1--12.
  Springer, 2003.

\bibitem[Guo et~al.(2007)Guo, Niedermeier, and Wernicke]{GuoNW07}
J.~Guo, R.~Niedermeier, and S.~Wernicke.
\newblock Parameterized complexity of vertex cover variants.
\newblock \emph{Theory Comput. Syst.}, 41\penalty0 (3):\penalty0 501--520,
  2007.

\bibitem[Lampis(2011)]{Lampis11}
M.~Lampis.
\newblock A kernel of order 2 k-c log k for vertex cover.
\newblock \emph{Inf. Process. Lett.}, 111\penalty0 (23-24):\penalty0
  1089--1091, 2011.

\bibitem[Lampis(2014)]{Lampis14}
M.~Lampis.
\newblock Parameterized approximation schemes using graph widths.
\newblock In \emph{{ICALP 2014} {(1)}}, volume 8572 of \emph{LNCS}, pages
  775--786. Springer, 2014.

\bibitem[Marx(2008)]{Marx08}
D.~Marx.
\newblock Parameterized complexity and approximation algorithms.
\newblock \emph{Comput. J.}, 51\penalty0 (1):\penalty0 60--78, 2008.

\bibitem[M{\"{o}}lle et~al.(2008)M{\"{o}}lle, Richter, and
  Rossmanith]{MolleRR08}
D.~M{\"{o}}lle, S.~Richter, and P.~Rossmanith.
\newblock Enumerate and expand: Improved algorithms for connected vertex cover
  and tree cover.
\newblock \emph{Theory Comput. Syst.}, 43\penalty0 (2):\penalty0 234--253,
  2008.

\bibitem[Niedermeier(2006)]{N06}
R.~Niedermeier.
\newblock \emph{Invitation to Fixed-Parameter Algorithms}.
\newblock Oxford University Press, 2006.

\bibitem[Niedermeier and Rossmanith(2003)]{Nied03}
R.~Niedermeier and P.~Rossmanith.
\newblock On efficient fixed-parameter algorithms for weighted vertex cover.
\newblock \emph{Journal of Algorithms}, 47\penalty0 (2):\penalty0 63 -- 77,
  2003.

\bibitem[Williamson and Shmoys(2011)]{WS11}
D.~P. Williamson and D.~B. Shmoys.
\newblock \emph{The Design of Approximation Algorithms}.
\newblock Cambridge University Press, New York, NY, USA, 1st edition, 2011.

\end{thebibliography}
\label{sec:biblio}

\end{document}